\documentclass[11pt]{article}
\usepackage{subfigure}
\usepackage{color}
\usepackage{url}
\usepackage{graphicx}
\usepackage{fullpage}

\newcommand{\qed}{\mbox{}\hspace*{\fill}\nolinebreak\mbox{$\rule{0.6em}{0.6em}$}
}

\definecolor{gray}{rgb}{0.5,0.5,0.5}

\newtheorem{theorem}{Theorem}[section]

\newtheorem{lemma}[theorem]{Lemma}
\newtheorem{claim}[theorem]{Claim}

\newtheorem{definition}{Definition}[section]

\newtheorem{proposition}[theorem]{Proposition}
\newenvironment{proof}{{\bf Proof:}}{$\qed$\par}

\usepackage{hyperref}
\hypersetup{
bookmarksnumbered
}
\usepackage{amsmath}

\title{Lower Bounds for Embedding into Distributions over Excluded Minor Graph
  Families
\footnote{An earlier version of this paper~\cite{cg:minor04} contained a small error
which we have fixed here.}
}
\author{Douglas E. Carroll
\thanks{
Douglas E. Carroll was at the department of Computer Science in UCLA when the
bulk of this research was conducted. Email:
{\tt dougecarroll@yahoo.com}}
\\University of California, Los Angeles
\and 
Ashish Goel
\thanks{
Departments of Management Science and Engineering  and (by courtesy)
Computer Science, Stanford University.
Email: {\tt ashishg@stanford.edu}. Research supported by grants from the
NSF and an Alfred P. Sloan faculty fellowship.}
\\Stanford University
}

\begin{document}

\maketitle

\begin{abstract}

It was shown recently by Fakcharoenphol {\em et al.}~\cite{FakcharRT03} 
that arbitrary finite metrics can be embedded into distributions over 
tree metrics with distortion $O(\log n)$. It is also known that this 
bound is tight since there are expander graphs which cannot be embedded 
into distributions over trees with better than $\Omega(\log n)$ distortion. 

We show that this same lower bound holds for embeddings into distributions
over any minor excluded family. Given a family of graphs $F$ which excludes
minor $M$ where $|M|=k$, we explicitly construct a family 
of graphs with treewidth-$(k+1)$ which cannot be embedded into a distribution 
over $F$ with better than $\Omega(\log n)$ distortion. Thus, while these 
minor excluded families of graphs are more expressive than
trees, they do not provide asymptotically better approximations in general.
An important corollary of this is that graphs of treewidth-$k$ cannot be
embedded into distributions over graphs of treewidth-$(k-3)$ with distortion
less than $\Omega(\log n)$.

We also extend a result of Alon {\em et al.}~\cite{AlonKPW95} by showing
that for any $k$, planar graphs cannot be embedded into distributions 
over treewidth-$k$ graphs with better than $\Omega(\log n)$ distortion.

\end{abstract}

\section{Introduction}

Many difficult problems can be approximated well when restricted to certain
classes of metrics~\cite{LinialLR95,Bartal98,Baker94}. Therefore, low distortion
embeddings into these restricted classes of metrics become very desirable. We
will refer to the original metric which we would like to embed as the {\em
  source metric} and the metric into which we would like to embed as the {\em
  target metric}.

Tree metrics form one such class of desirable target metrics in the sense
that many difficult problems become tractable when restricted to
tree metrics. However, they are not sufficiently expressive; i.e. it
has been shown that there are classes of metrics which cannot be
approximated well by trees. In particular, Rabinovich and Raz
\cite{RabinovichR98} have proved that graph metrics cannot be embedded
into trees with distortion better than $\mbox{girth}/3-1$.

Therefore, subsequent approaches have proposed the use of more expressive
classes of metrics. Alon {\em et al.} \cite{AlonKPW95} showed that any
n-point metric can be embedded with $ 2^{O(\sqrt{\log n \log \log n })}$
distortion into distributions over spanning trees. In so doing they
demonstrated that such probabilistic metrics are more expressive than tree
metrics. Bartal \cite{Bartal96} formally defined probabilistic embeddings 
and proposed using distributions over arbitrary
dominating tree metrics. He showed that any finite metric could be embedded
into distributions over trees with $O({\log}^2 n)$ distortion. He subsequently
improved this bound to $O(\log n \log \log n)$ \cite{Bartal98}.

This path culminated recently in the result of Fakcharoenphol {\em et al.}
\cite{FakcharRT03} in which they improved this bound to $O(\log n)$
distortion. This upper bound is known to be tight since there exist graphs
which cannot be embedded into such distributions with better than $\Omega(\log
n)$ distortion. This lower bound follows naturally from the fact that any
distribution over trees can be embedded into $\ell_1$ with constant
distortion, and the existence of expander graphs which cannot be embedded into
$\ell_1$ with distortion better than $\Omega(\log n)$. Surprisingly, Gupta
{\em et al.} \cite{GuptaNRS99} showed that this same bound is in fact achieved
by source graph metrics of treewidth-$2$. Their result is also implicit in the
work of Imase and Waxman~\cite{ImaseW91} where they establish a lower bound on
the competitive ratio of online Steiner trees.

It is plausible to hypothesize that there are more general and expressive
classes of target metrics. We explore using distributions over minor closed
families of graphs and show that asymptotically, they are no stronger than
distributions over trees. We show a $\Omega(\log n)$ lower bound on the distortion
even when the source metrics are graphs of low treewidth. More precisely,
for any minor $M$, where $|M|=k$, we exhibit a construction for an
infinite class of finite metrics of treewidth-$(k+1)$ for which any embedding
into distributions over families of graphs excluding $M$ achieves $\Omega(\log
n)$ distortion. A corollary of this fact is that treewidth-$k$ graphs cannot
be embedded into distributions over treewidth-$(k-3)$ graphs with distortion
less than $\Omega(\log n)$.

A weaker result can be inferred directly from Rao \cite{Rao99} who proved
that any minor closed family can be embedded into
$\ell_1$ with distortion $O(\sqrt{\log n})$. Consequently, the expanders
exhibited by Linial {\em et al.} \cite{LinialLR95} cannot be embedded into
distributions over minor closed families with distortion less than
$\Omega(\sqrt{\log n})$.

One can derive a lower bound of $\Omega(\log n)$ by combining the method
of Klein {\em et al.} \cite{KleinPR93} for decomposing minor excluded graphs
with Bartal's \cite{Bartal96} proof of the lower bound for probabilistic 
embeddings into trees. This bound also follows from the recent paper
of Rabinovich on the average distortion of embeddings into $\ell_1$.
However, we show this same
bound holds for embeddings of simple, low-treewidth source metrics.


We continue by exploring the difficulty of embedding planar graph metrics into
distributions over other minor excluded graph families. Alon {\em et al.}
showed that any embedding of a 2-dimensional grid into a distribution
over spanning subtrees must have distortion $\Omega(\log n)$. We show
that for any fixed $k$, the same lower bound holds for embeddings of
2-dimensional grids into distributions over dominating treewidth-$k$ graph
metrics.

Note that our $\Omega(\log n)$ lower bounds hide polynomial factors of $k$.

\subsection{Techniques}

We employ Yao's MiniMax principle to prove both lower bounds -- it suffices to
show that for some distribution over the edges of the source graph, any
embedding of the source graph into a dominating target graph has large
expected distortion. In both cases the expected distortion is shown to be
logarithmic in the number of vertices in the graph.

For the main result, we show that given any minor $M$, one can recursively
build a family of source graphs which guarantee a large expected distortion
when embedded into a graph which excludes $M$ as a minor. This structure is
detailed in section 3. While we use some elements from the construction of
Gupta~{\em et al.}~\cite{GuptaNRS99}, the main building block in our recursive
construction is a stretched clique rather than the simpler treewidth-2 graphs
used by Gupta~{\em et al.}, and hence, our proof is more involved. As a step
in our proof, we show that if graph $G$ does not contain $H$ as a minor, then
subdivisions of $H$ can not be embedded into $G$ with a small distortion.

To show the lower bound for embedding planar graphs into bounded tree-width
graphs, we use the notion of nice tree
decompositions~\cite{Bodlaender98,Kloks94} of bounded tree-width graphs.  We
show that for a range of set sizes, the subsets of a $2$-dimensional grid have
much larger separators than do their embeddings in the nice tree
decompositions; this suffices to prove an $\Omega(\log n)$ lower bound with a
little more work. This argument is presented in section 4.

An earlier version of this paper~\cite{cg:minor04} had a small error. We were
using a recursive construction for a graph where there were many disjoint
paths between a source-sink pair. However, they were not all shortest paths, a
fact that we overlooked. This version of the paper corrects that error. The
correction is relatively simple: the only substantive change is that we use
$K_{n,n}$ as opposed to $K_n$ in the recursive construction.

\section{Definitions and Preliminaries}

Given two metric spaces $(G,\nu)$ and $(H,\mu)$ and an embedding
$\Phi:G \rightarrow H$, we say that the {\em distortion} of the embedding is
$\|\Phi\|\cdot\|\Phi^{-1}\|$ where 
\begin{eqnarray*}
\|\Phi\| &=& \max_{x,y \in G} \frac{\mu(\Phi(x),\Phi(y))}{\nu(x,y)}, \\
\|\Phi^{-1}\| &=& \max_{x,y \in G} \frac{\nu(x,y)}{\mu(\Phi(x),\Phi(y))}
\end{eqnarray*}

and $(G,\nu)$ {\em $\alpha$-approximates} $(H,\mu)$ if the {\em distortion} is
no more than $\alpha$. We say that $\mu$ {\em dominates} $\nu$ if 
$\mu(\Phi(x),\Phi(y)) \geq \nu(x,y)$  $\forall x,y$.

\begin{definition}
Given a graph $G=(V_G,E_G)$, a tree $T=(V_T,E_T)$ and a 
collection $\{X_i | i \in V_T\}$ of subsets of $V_G$, then $(\{X_i\},T)$ is said
to be a \emph{tree decomposition} of G if 
\begin{enumerate}
\item   $\underset{i \in V_T} \bigcup X_i = V_G$,
\item   for each edge $e \in E_G$, there exists $i \in V_T$ such 
        that the endpoints of $e$ are in $X_i$, and
\item   for all $i,j,k \in V_T:$ if $j$ lies on the path from $i$ to $k$ in $T$,
        then $X_i \cap X_k \subseteq X_j$.
\end{enumerate}
The {\em width} of a tree decomposition is $\underset {i \in V_T} {\max} |X_i|-1$.
The {\em treewidth} of a graph $G$ is the minimum width over all tree 
decompositions of $G$.
\end{definition}

\begin{definition}
A tree decomposition $(\{X_i\},T=(V_T,E_T))$ of a graph $G=(V_G,E_G)$
is said to be a {\em nice} decomposition if:
\begin{enumerate}
\item   $T$ is a rooted binary tree
\item   if node $i \in V_T$ has two children $j_1,j_2$,
        then $X_{j_1}=X_{j_2}=X_i$
\item   if node $i \in V_T$ has one child $j$,
        then either $X_i \subset X_j$ and $|X_i|=|X_j|-1$ or
        $X_j \subset X_i$ and $|X_j|=|X_i|-1$
\item   if node $i \in V_T$ is a leaf, then $|X_i|=1$ 
\end{enumerate}
\end{definition}

\begin{proposition}
\label{prop:nice}
\cite{Bodlaender98,Kloks94}
If graph $G$ has a tree decomposition of width $k$, then $G$ has a
nice tree decomposition of width $k$.
\end{proposition}

The following property follows directly from the definition of a tree
decomposition:

\begin{proposition}
\label{prop:path}
For all $i,j,k \in V_T:$ if $j$ lies on the path from $i$ to $k$ in $T$, 
and $x_1 \in X_i$ and $x_2 \in X_k$ then any path connecting $x_1$ 
and $x_2$ in $G$ must contain a node in $X_j$.
\end{proposition}

For a more complete exposition of treewidth see \cite{Kloks94,Bodlaender98}.

\begin{definition}
Given a graph $H$, a {\em $k$-subdivision} of $H$ is defined as the
graph resulting from the replacement of each edge by a path of length
$k$.
\end{definition}

We will use $G_{n,m}$ to denote the $m$-subdivision of $K_{n,n}$.
We observe that $treewidth(G_{n,m}) = n =treewidth(K_{n,n})$.

\section{Distributions over Minor Excluded Families}

\subsection{Outline}

Given a family $F$ which excludes minor $M$ where $|M|=k$, we recursively
construct a family of graphs $H_i$ which embed into any graph in $F$ with 
average distortion $\Omega(\log n)$ where $n=|H_i|$. Then by Yao's MiniMax
principle, any embedding into a distribution over $F$ must have distortion
$\Omega(\log n)$ as claimed.

\subsection{Results}

Before proceeding with our main result, we need the following technical
lemma\footnote{ This lemma appears to be folklore. But we have not been able
  to find a proof in the literature. Our proof is non-trivial and might be
  useful in other contexts, so we have provided the proof in this paper.}.
\begin{lemma}
  \label{lem:minor}
  Let $G,H$ be graphs such that $G$ does not contain $H$ as a minor and let
  $J$ be the $k$-subdivision of $H$. If $f$ is an embedding of the metric of
  $J$ into the metric of $G$, then $f$ has distortion at least $k/6-1/2$.
\end{lemma}
  \newcommand{\wa}{w_{\lfloor k/2\rfloor}}
  \newcommand{\wb}{w_{\lfloor k/2\rfloor+1}}
\begin{proof}
  
  We will use the linear extension framework of Rabinovich and
  Raz~\cite{RabinovichR98}. While we borrow some of their techniques, we have
  not found a reduction that would use their lower bounds (on the distortion
  of an embedding of a graph into another with a smaller Euler characteristic)
  as a black-box to prove this technical lemma.
  
  The linear extension $\tilde{Q}$ of a graph $Q$ is obtained by identifying
  with each edge of $Q$ a line of length 1. All the points on all the lines
  belong to $\tilde{Q}$. The metric $d_Q$ can be extended to the metric
  $\tilde{d}_Q$ in the following natural fashion.  If $x$ and $y$ lie on the
  same line $(u,v)$ then $\tilde{d}_Q(x,y)$ is merely their distance on the
  line.  If $x$ lies on the line $(u,v)$ and $y$ lies on a different line
  $(w,r)$, then
  \begin{eqnarray*}
    \tilde{d}_Q(x,y) = \min\{&&\tilde{d}_Q(x,u) + d_Q(u,w) + \tilde{d}_Q(w,y), \\
    && \tilde{d}_Q(x,v) + d_Q(v,w) + \tilde{d}_Q(w,y),\\
    && \tilde{d}_Q(x,u) + d_Q(u,r) + \tilde{d}_Q(r,y), \\
    &&\tilde{d}_Q(x,v) + d_Q(v,r) + \tilde{d}_Q(r,y)\}
  \end{eqnarray*}
  We will refer to the original vertices as vertices of $\tilde{Q}$ and to the
  original edges of $Q$ as the lines of $\tilde{Q}$.
  We can now extend the embedding $f$ into a continuous map
  $\tilde{f}:\tilde{J}\rightarrow\tilde{G}$ such that 
  \begin{enumerate}
  \item $\tilde{f}$ agrees with
  $f$ on vertices of $\tilde{J}$, and
\item if $x\in (u,v)$ where $(u,v)$ is a line of $\tilde{J}$ then
  $\tilde{f}(x)$ lies on a shortest path from $\tilde{f}(u)$ to $\tilde{f}(v)$
  in $\tilde{G}$ such that $\tilde{d}_J(x,u)/\tilde{d}_J(x,v) =
  \tilde{d}_G(\tilde{f}(x),\tilde{f}(u))
  /\tilde{d}_G(\tilde{f}(x),\tilde{f}(v))$.
  \end{enumerate}
  
  Since $\tilde{f}$ is continuous, the entire line $(u,v)$ in $\tilde{J}$ must
  be mapped to a single shortest path from $\tilde{f}(u)$ to $\tilde{f}(v)$ in
  $\tilde{G}$. We will now assume that the distortion $\alpha$ of $f$ is less
  than $k/6-1/2$ and derive a contradiction. But first, we need to state and
  prove the following useful claim:

  \begin{claim}
    If the points $\tilde{f}(x)$ and $\tilde{f}(y)$ lie on the same line in
    $\tilde{G}$, the points $x,x'$ lie on the same line in $\tilde{J}$, and
    the points $y,y'$ lie on the same line in $\tilde{J}$, then
    $\tilde{d}_J(x',y') \le 2\alpha+1$.
  \end{claim}
  \begin{proof}
    Suppose $x,x'$ and $y,y'$ lie on lines $(p,q)$ and $(r,s)$ in $\tilde{J}$,
    respectively. Use $X$ to denote the quantity
    $\tilde{d}_G(\tilde{f}(x),\tilde{f}(y))$ and $Y$ to denote the quantity
    $\tilde{d}_G(\tilde{f}(p), \tilde{f}(x)) + \tilde{d}_G(\tilde{f}(x),
    \tilde{f}(q)) + \tilde{d}_G(\tilde{f}(r), \tilde{f}(y)) +
    \tilde{d}_G(\tilde{f}(y), \tilde{f}(s))$.  Since $\tilde{f}(x)$ and
    $\tilde{f}(y)$ lie on the shortest paths from $\tilde{f}(p)$ to
    $\tilde{f}(q)$ and from $\tilde{f}(r)$ to $\tilde{f}(s)$, respectively, we
    have $\tilde{d}_G(\tilde{f}(p), \tilde{f}(q)) + \tilde{d}_G(\tilde{f}(r),
    \tilde{f}(s)) = Y$. Also, by triangle inequality, we have
    $\tilde{d}_G(\tilde{f}(p), \tilde{f}(r)) + \tilde{d}_G(\tilde{f}(p),
    \tilde{f}(s)) + \tilde{d}_G(\tilde{f}(q), \tilde{f}(r)) +
    \tilde{d}_G(\tilde{f}(q), \tilde{f}(s)) \le 4X + 2Y$. Clearly, $X \le 1$
    and $Y \ge 2$, so we have
    $$
    \frac{
      \tilde{d}_G(\tilde{f}(p), \tilde{f}(r)) + \tilde{d}_G(\tilde{f}(p),
      \tilde{f}(s)) + \tilde{d}_G(\tilde{f}(q), \tilde{f}(r)) +
      \tilde{d}_G(\tilde{f}(q), \tilde{f}(s))
      }
    {
      \tilde{d}_G(\tilde{f}(p), \tilde{f}(q)) + \tilde{d}_G(\tilde{f}(r),
      \tilde{f}(s))
      } \le 4.
    $$
    Since the distortion is $\alpha$, we must have
    $$
    \frac{
      \tilde{d}_J(p,r) +
      \tilde{d}_J(p,s) + \tilde{d}_J(q,r) + \tilde{d}_J(q,s)
      }
    {
      \tilde{d}_J(p,q) + \tilde{d}_J(r,s)
      }
    \le 4\alpha.
    $$
    
    But $\tilde{d}_J(p,q) = \tilde{d}_J(r,s) = 1$. Also, $\tilde{d}_J(p,r) +
    \tilde{d}_J(p,s) + \tilde{d}_J(q,r) + \tilde{d}_J(q,s) \ge
    4 \tilde{d}_J(x',y') - 4$.
    Hence, we have $(4 \tilde{d}_J(x',y') - 4)/2 \le 4\alpha$, or
    $\tilde{d}_J(x',y') \le 2\alpha + 1.$
  \end{proof}
  
  We will now continue with the proof of Lemma~\ref{lem:minor}.  For any edge
  $(u,v)\in H$, consider the $u$-$v$ path of length $k$ in $\tilde J$.
  Consider the image of this path in $\tilde G$ under $\tilde f$. The image
  need not be a simple path, but must contain a simple path from $\tilde f(u)$
  to $\tilde f(v)$. We choose any such simple path arbitrarily, and call it
  the representative path for $(u,v)$, and denote it by $P(u,v)$. Start
  traversing this path from $\tilde f(u)$ to$\tilde f(v)$ and stop at the
  first vertex $q$ which is the image of a point $x$ on the $u$-$v$ path in
  $\tilde J$ such that $\tilde d_J(u,x) \geq k/2-\alpha-1/2$. Let $(p,q)$ be
  the last edge traversed. We will call this the representative line of
  $(u,v)$ and denote it by $L(u,v)$. Consider a point $y$ on the $u$-$v$ path
  in $\tilde J$ that maps to $p$. The choice of $p$ and $q$ implies that
  $\tilde d_J(u,y) < k/2-\alpha-1/2$, and hence, we can apply claim 1 to
  conclude that $ k/2-\alpha-1/2 \le d_J(u,x) \le k/2+\alpha+1/2$. Now, for
  any point $z$ not on the $u$-$v$ path in $\tilde J$, we have $\tilde
  d_J(x,z) \ge k/2 - \alpha -1/2$. Since we assumed $\alpha < k/6 -1/2$, we
  have $\tilde d_J(x,z) > 2\alpha + 1$ and hence, $\tilde f(z)$ can not lie on
  the line $L(u,v)$.
  
  Thus the representative line $L(u,v)$ has two nice properties: it lies on a
  simple path from $\tilde f(u)$ to $\tilde f(v)$ and does not contain the
  image of any point not on the $u$-$v$ path in $\tilde J$.
  
  Now, we perform the following simple process in the graph $G$: For every
  edge $(u,v)$ of $H$, contract all the edges on the representative path
  $P(u,v)$ in $G$ except the representative edge $L(u,v)$.  Since the
  representative line of an edge $(u,v)$ of $H$ does not intersect the
  representative path of any other edge of $H$, no representative edge gets
  contracted. The resulting graph is a minor of $G$, but also contains $H$ as
  a minor, which is a contradiction. Hence $\alpha \geq k/6 - 1/2$.

\end{proof}

Now we can prove our main result:
\begin{theorem}
\label{thm:minor}
Let $F$ be a minor closed family of graphs which excludes minor $M$ where
$|V_M| = n$. There exists an infinite family of graphs ${H_i}$ with
treewidth-$(n+1)$ such that any $\alpha-$approximation of the metric of $H_i$
by a distribution over dominating graph metrics in $F$ has
$\alpha=\Omega(\log |H_i|)$.
\end{theorem}

\begin{proof}

We proceed by constructing an infinite sequence ${H_i}$ of
graphs of treewidth-$(n+1)$ and show that the minimum distortion with which they can
be embedded into a distribution over graph metrics in $F$ grows
without bound as $i$ increases.

First we construct $H_1$: Construct the graph $H_1$ by taking $K_{n,n}$ and
attaching each vertex $l$ in the ``left half'' of $K_{n,n}$ to source $s$ with
$n$ disjoint paths of length $n$ and attaching each vertex $r$ in the ``right
half'' of $K_{n,n}$ to sink $t$ with $n$ paths of length $n$. Call $s$ and $t$
the terminals of $H_1$.

The graph $H_1$ has $m = 2n^3 +n^2$ edges. We now show that this graph
contains exactly $n^2$ edge-disjoint $s,t$ paths of length $2n+1$.  Label the
$s$ to $l_i$ paths $sp_{i,1}$ to $sp_{i,n}$ and the $r_i$ to $t$ paths
$tp_{i,1}$ to $tp_{i,n}$.  Observe that the paths formed as follows are edge
disjoint:

\bigskip path $sp_{i,j}$ followed by edge $(l_i,r_j)$ followed by path
$tp_{j,i}$

\bigskip
\noindent

The $H_i$ are constructed recursively: For each $i$, construct the graph $H_i$
by replacing every edge in $H_1$ with a copy of $H_{i-1}$.  Therefore, $H_i$
has $m^i = (2n^3 + n^2)^i$ edges and the two vertices at the end of any edge
in the original $H_1$ are connected by $n^{2(i-1)}$ edge-disjoint paths of
length $(2n+1)^{i-1}$.  Also note that $H_i$ has treewidth $\leq n+1$ (for
completeness, we provide a proof in lemma~\ref{lem:width}).

As in \cite{GuptaNRS99} we use Yao's MiniMax Principle to prove the lower
bound.  We show that there is a distribution $d$ over the edges of $H_i$ such
that for any embedding of $H_i$ into a graph which excludes minor $M$, an edge
chosen randomly from distribution $d$ has an expected distortion of
$\Omega(i)$. Then by Yao's MiniMax principle, for any embedding of $H_i$ into
a distribution over graphs in $F$, there must be an edge with expected
distortion $\Omega(i)$. We shall assume a uniform distribution over the edges.

Let $U$ be a graph which excludes minor $M$ and let $\Phi:H_i \rightarrow U$
be an embedding of $H_i$ into $U$ such that distances in $U$ dominate their
corresponding distances in $H_i$. For each edge $e \in H_i$ we will give $e$
the color $j, 1 \leq j \leq i-1$ if $\Phi$ distorts $e$ by at least
$\frac{1}{6}(2n+1)^j - 1/2$. Note that an edge may have many colors.

Consider the copies of $G_{n,(2n+1)^{i-1}}$ in $H_i$. $H_i$ contains a copy of
$K_{n,n}$ in which every edge has been replaced with a copy of $H_{i-1}$.
Each copy of $H_{i-1}$ has $n^{2(i-1)}$ edge disjoint paths of length
$(2n+1)^{i-1}$. Thus, $H_i$ clearly contains at least $n^{2(i-1)}$ edge
disjoint copies of $G_{n,(2n+1)^{i-1}}$.

$U$ does not contain $M$ as a minor. But since $|V_M| = n$, $K_n$ contains $M$
as a minor. Also, $K_{n,n}$ contains $K_n$ and hence $M$ as a minor, which in
turn implies that $U$ does not contain $K_{n,n}$ as a minor. Thus, by
lemma~\ref{lem:minor}, at least one edge in each copy of $G_{n,(2n+1)^{i-1}}$
has distortion $\geq \frac {1}{6}(2n+1)^{i-1} - 1/2$ and hence has color
$i-1$. Since $H_i$ comprises $m$ copies of $H_{i-1}$, it contains $m
(n^2)^{i-2}$ copies of $G_{n,(2n+1)^{i-2}}$. Therefore, there are at least $m
(n^2)^{i-2}$ edges with color $i-2$. In general, there are at least $\geq
m^{i-j-1} \cdot n^{2j}$ edges with color $j$.

The distortion of an edge is $\geq \frac {1}{6}(2n+1)^{j} - 1/2$ where $j$ is
the largest of the edge's colors. For each edge $e \in E_{H_i}$ let $C_e$ be
the set of colors which apply to edge $e$. Clearly, $\forall e \in E_{H_i}$,

\begin{eqnarray*}
  \underset{j \in C_e}{\sum} (\frac{1}{6}(2n+1)^{j} - 1/2) &\leq& 2 \cdot
  \underset{j \in C_e}{\max} (\frac{1}{6}(2n+1)^{j} - 1/2)
\end{eqnarray*}

Thus,

\begin{eqnarray*}
\underset{e \in E_{H_i}}{\sum}\underset{j \in C_e}{\max} (\frac {1}{6}
(2n+1)^{j} - 1/2)
&\geq&
\frac{1}{2}\underset{e \in E_{H_i}}{\sum} \underset{j \in C_e}{\sum} (\frac{1}{6}(2n+1)^{j} - 1/2)\\
&\ge& \frac{1}{18}\left(\overset{i-1}{\underset{j=1}{\sum}} |\{e|j \in C_e\}|\cdot
(2n+1)^{j}\right) - m^i/4\\
&& \ \ \ \ \ \ \ \ \ \ \mbox{[Separating out the -1/2 term for $j=1$]}\\
&\geq& \frac{1}{18}\left(\overset{i-1}{\underset{j=1}{\sum}} m^{i-j-1} \cdot n^{2j}
\cdot (2n+1)^{j}\right) -m^i/4 \\
&=& \frac{1}{18}\left(\overset{i-1}{\underset{j=1}{\sum}} m^{i-1}\right) - m^i/4 \\
&=& \frac{1}{18} \left((i-1)\cdot m^{i-1}\right) - m^i/4\\
&=& (m^{i}/4)\left(\frac{i-1}{4.5 m} - 1\right)\\
&=& m^i\Omega(i/m)
\end{eqnarray*}

Then, since $H_i$ has $m^i$ edges, there must be at least one edge with
distortion $\Omega(i/m) = \Omega(\log|H_i|)$, ignoring polynomial factors of
$n$.

\end{proof}

\begin{lemma}
  \label{lem:width} The graph $H_i$ has treewidth at most $n+1$.
\end{lemma}
\begin{proof}
  For $i>1$ will use the terms source and sink vertices of $H_i$ to refer to
  the source and sink of the copy of $H_1$ in which each edge was replaced by
  a copy of $H_{i-1}$. We will also assume that $n \ge 1$ since the case $n=0$
  is not well defined.  We will prove the lemma by induction. In fact we will
  prove a slightly stronger result: that the graph $H_i$ has a decomposition
  of treewidth at most $n+1$ such that one of the supernodes in the tree
  decomposition has both the source and the terminal vertex of $H_i$.
  
  As a base case, consider the following decomposition of $H_1$. Let $l_1,
  l_2, \ldots, l_n$ refer to the ``left'' vertices in the copy of $K_{n,n}$
  within $H_1$, i.e., vertices in $K_{n,n}$ which are connected to the source
  $s$. Similarly, let $r_1, r_2, \ldots, r_n$ refer to the ``right'' vertices,
  i.e., vertices in $K_{n,n}$ which are connected to the sink $t$. Consider a
  decomposition which has a central supernode containing $s,t$ and all the
  left vertices. This supernode is connected via separate edges to supernodes
  $R_j$, $j=1, 2, \ldots, n$, where $R_j$ consists of all the left vertices,
  the sink $t$, and the right vertex $r_j$. Notice that each of the supernodes
  described so far has exactly $n+2$ nodes. Notice also that $s,t$ appear
  together in the central supernode. All the edges in $K_{n,n}$ are already
  ``covered'' i.e. there is a supernode which contains both endpoints. The
  edges on the paths from the source to the left vertices and the right
  vertices to the sinks are not covered. But observe that for each left vertex
  $l_j$, there is a supernode (the central supernode) which contains both
  $l_j$ and $s$. Also, for each right vertex $r_j$, there is a supernode
  (specifically, $R_j$) which contains both $r_j$ and $t$. Consider any path
  $P$ from source $s$ to node $l_j$. Obtain a tree decomposition of this path
  (with treewidth 1), and for each supernode in this decomposition, add node
  $l_j$ to the supernode, if not already present; this new decomposition now
  has treewidth 2. We call this the augmented decomposition for $P$. This
  augmented decomposition must have at least one supernode which contains both
  $l_j$ and $s$; connect an arbitrary supernode which contains both $s$ and
  $l_j$ in the augmented decomposition to the central supernode. All edges on
  the path $P$ are now covered. Repeat this process for all paths from $s$ to
  the left vertices. Repeat the same process for every path from the right
  vertices to $t$ except that an arbitrarily chosen supernode with both $t$
  and $r_j$ in the augmented decomposition for the path is connected to
  supernode $R_j$ rather than to the central supernode. The resulting tree
  decomposition has treewidth $\max\{2,n+1\} = n+1$. This completes the base
  case for the induction. It is easy to see that the resulting decomposition
  satisfies all the properties required for a tree-width of $n+1$.

  For the induction step, assume the hypothesis is true for $i \le k-1$. The
  induction step will mimic the recursive construction of $H_k$. Consider
  $i=k$. Consider the tree decomposition $D_1$ of $H_1$ as described
  above. Then, for each edge $e$ of $H_1$, do the following.
  \begin{enumerate}
  \item First, take a fresh copy of $H_{k-1}$, with source $s_e$ and
    destination $t_e$. Identify $s_e$ and $t_e$ with the endpoints of $e$.
  \item Take a tree decomposition of $H_{k-1}$ as guaranteed above; call this
    $D_{k-1,e}$. This copy must have a supernode, say $A$, with both the
    source $s_e$ and the destination $t_e$ of $H_{k-1}$.
  \item There must be a supernode, say $B$, in $D_1$ which contains both
    endpoints of $e$. Join $A$ and $B$ with an edge.
  \end{enumerate}
  It is easy to see that the resulting decomposition satisfies all the
  properties required for a tree-width of $n+1$. Also, the decomposition $D_1$
  has a supernode with both $s$ and $t$, which completes the proof of this
  lemma.
\end{proof}

\section{Planar Graphs}

\subsection{Outline}


First we show that given a $2$-dimensional grid, there is a distribution over
the edges such that any embedding into a
treewidth-$k$ graph metric has high expected distortion. The proof builds on
the work of Alon {\em et al.} \cite{AlonKPW95}. By Yao's MiniMax principle,
this is enough to show that the $2$-dimensional grid can not be embedded into
a distribution over such graphs with distortion less than $\Omega(\log n)$.

\subsection{Results}

In this section we will use $GRID_n$ to denote the planar graph
consisting of the $n \times n$ grid.

\begin{lemma}\label{lem:rows}(From~\cite {AlonKPW95})
If $A$ is a set of $\beta^2$ vertices in $GRID_n$, where $\beta^2 \leq \frac {n^2}{2}$,
then there are at least $\beta$ rows that $A$ intersects but does not fill or at least 
$\beta$ columns that $A$ intersects but does not fill.
\end{lemma}

\begin{lemma}\label{lem:vertices}(Modified from~\cite {AlonKPW95})
If $A$ is a set of $\beta^2$ vertices in $GRID_n$, where $\beta^2 \leq \frac {n^2}{2}$,
and $B$ is a set of at most $\beta / 4$ vertices in $A$, then there are at least $\beta/2$
vertices in $A$ that have neighbors outside $A$ and have distance at least $\frac {\beta}{4|B|}$
from each vertex of $B$.
\end{lemma}

\begin{proof}
By Lemma~\ref{lem:rows}, there is a set $C$ of at least $\beta$ vertices in $A$ which
are in distinct rows or distinct columns and have neighbors outside of $A$. Since
they are in distinct rows, a vertex of $B$ can be at distance $< \frac {\beta}{4|B|}$
of at most $\frac{\beta}{2|B|}$ vertices in $C$. Thus, there are at least 
$\frac{\beta}{2}$ vertices at distance at least $\frac {\beta}{4|B|}$ from each
vertex in $B$.
\end{proof}

\begin{lemma}\label{lem:dist}
Let $H$ be a graph of treewidth $k$, $\Phi:GRID_n \rightarrow H$ be an embedding of $GRID_n$
into $H$, and $\beta \leq n/4$. Then there are at least $\frac {n^2}{24\beta}$ edges $(u,v)$
such that $d_H(u,v)>\frac{\beta}{16(k+1)}$
\end{lemma}

\begin{proof}
Since $H$ has treewidth-$k$, it must have a tree decomposition of width $k$. 
Moreover, it must have a {\em nice} tree decomposition $({X_i},T)$ of width $k$ 
by Proposition~\ref{prop:nice}. 

Given a subtree $S$ of $T$, define the $\mbox{\sc{HSize}}(S)=|\underset{i \in V_S}{\bigcup} X_i|$. 
Since $({X_i},T)$ is a nice decomposition, we know that $T$ has a maximum degree
of $3$. We also know that for any two adjacent vertices $i,j \in T$, $|X_i-X_j| \leq 1$.
Thus, every subtree $S$ of $T$ has an edge $e$ such that 
removing $e$ creates $2$ subtrees each with $\mbox{\sc{HSize}}$ at least $1/3 \cdot \mbox{\sc{HSize}}(S)$.
Note that if $S_1$ and $S_2$ are the two subtrees of $T$, $\underset{i \in V_{S_1}}{\bigcup} X_i$
and $\underset{i \in V_{S_2}}{\bigcup} X_i$ are not disjoint.

Start with $T$ and successively delete edges from the remaining component with
the largest $\mbox{\sc{HSize}}$ such that the $\mbox{\sc{HSize}}$s of the resulting subtrees are as
evenly divided as possible. Do this until $\lceil \frac{n^2}{3\beta^2} \rceil -1$ edges
have been deleted and there are $\lceil \frac{n^2}{3\beta^2} \rceil$ pieces. The
smallest piece will always be at least $1/3$ the $\mbox{\sc{HSize}}$ of the previous largest
piece. Therefore, on the average these pieces have $\mbox{\sc{HSize}}=3 \beta^2$ and the
smallest will have $\mbox{\sc{HSize}} \geq \beta^2$.

Since each deleted edge of $T$ is incident to $2$ pieces, the average number of
pieces incident with a piece is less than $2$. Thus, at least half the pieces are
incident with no more than $4$ edges.

Each deleted edge in $T$ represents a set of points which form a vertex cut of 
$H$ of size $\leq k+1$. Thus, there are $\frac{n^2}{6\beta^2}$ pieces of $\mbox{\sc{HSize}}$
$\geq \beta^2$ which are separated from the rest of $H$ by a cut of size
$\leq 4(k+1)$. Let $A$ be a piece of $\mbox{\sc{HSize}} \geq \beta^2$ and let $B$ be the subset of
size $\leq 4(k+1)$ separating $A$ from the rest of $H$. Then by Lemma~\ref{lem:vertices}, $A$ has
at least $\beta / 2$ vertices with neighbors outside of the piece whose distance
from the vertices of $B$ is at least $\frac{\beta}{16(k+1)}$.
Thus, there are at least $\frac{n^2}{24\beta}$ edges which are each distorted
by a factor of $\frac{\beta}{16(k+1)}$.
\end{proof}

\begin{theorem}
\label{thm:grid}
Any $\alpha-$approximation of the metric of $GRID_n$ by a distribution over
dominating treewidth $k$ graphs has $\alpha=\Omega(\log n)$.
\end{theorem}

\begin{proof}
Let H be an arbitrary graph of treewidth $k$ whose metric dominates that of
$GRID_n$. By Lemma~\ref{lem:dist}, there are at least $\frac {n^2}{24\beta}$ edges 
which are distorted by $>\beta/16(k+1)$ for any $\beta \leq \frac{n}{4}$. 
Let $X$ be the distortion of an edge chosen uniformly at random from
$GRID_n$. $X$ can only take on non-negative integer values, so the
expected distortion is
$E[X]=\underset{i \geq 1}{\sum}Prob(X \geq i)$. For $i \leq \frac{n}{64(k+1)}$,
let $\beta = 16(k+1) i$. Then,

\begin{eqnarray*}
E[X] &=& \underset{i \geq 1}{\sum}Prob(X \geq i)\\
&>& \underset{i \geq 1}{\overset {\lfloor n/64(k+1) \rfloor}{\sum}} \frac {n^2}{24 \cdot 16(k+1)i \cdot 2n(n-1)}\\
&>& \underset{i \geq 1}{\overset {\lfloor n/64(k+1) \rfloor}{\sum}} \frac {1}{2 \cdot 24 \cdot 16(k+1)i} \\
&=& \frac {1}{768(k+1)} \underset{i \geq 1}{\overset {\lfloor n/64(k+1) \rfloor}{\sum}} \frac {1}{i}\\
&=& \Omega (\log n)
\end{eqnarray*}

Since $H$ was arbitrarily chosen, then by Yao's MiniMax principle if $GRID_n$
is embedded into a distribution over treewidth-$k$ graphs, there must be an
edge with expected distortion of $\Omega (\log n)$.
\end{proof}

\section{Conclusions}

It is interesting to note that
the inability of minor closed families to approximate all graphs well
supports the conjecture~\cite{GuptaNRS99}~\cite{ChekuriGNRS03} that
minor closed families of graphs (and therefore distributions over such
families) can be embedded into $\ell_1$ with constant distortion. Since Linial
{\em et al.}~\cite{LinialLR95} showed a lower bound of $\Omega(\log n)$ for
embedding arbitrary graph metrics into $\ell_1$, the conjecture further
implies that there must be families of graphs which cannot be embedded into
distributions over excluded minor graphs with distortion less than
$\Omega(\log n)$.

However, the particular inability of minor closed families 
to approximate other minor closed families also eliminates one
potential approach to embedding these families into $\ell_1$: 
Gupta {\em et al.}~\cite{GuptaNRS99} showed that although treewidth-$2$
graphs and treewidth-$1$
graphs are both embeddable into $\ell_1$ with constant distortion,
treewidth-$2$ graphs are not embeddable into distributions over
treewidth-$1$ graphs with constant distortion. We have shown that a
similar structure holds for all higher treewidths.  Thus, an
approach which attempts to repeatedly embed bounded treewidth graphs
into (distributions over) graphs with lower treewidth will not work.

\bigskip
\noindent
\textbf{Acknowledgements:} The authors would like to thank an anonymous
reviewer for pointing out the aforementioned folklore result. We would also
like to thank Elias Koutsoupias for valuable discussions, and Adam Meyerson
and Shailesh Vaya for helpful comments on previous drafts. The error mentioned
in the introduction was pointed out to us by Alex Jaffe and his advisor James
Lee from the University of Washington. Along with their coauthors, they also
discovered a fix~\cite{cjlv:minors08}. We discovered the corrected proof after
we were informed of the bug by Alex and James but without seeing their paper
or knowing any particulars of their proof. We would like to thank the authors
of~\cite{cjlv:minors08} for informing us of the bug and giving us an
opportunity to correct it and publish this corrected version before their
result got published.


\begin{thebibliography}{99}

\bibitem{AlonKPW95}
N. Alon, R. Karp, D. Peleg, and D. West, ``A Graph-Theoretic Game and its 
Application to the k-Server Problem'',
{\em SIAM J. Comput.}, 24:1 (1998), pp. 78-100.

\bibitem{Bartal96}
Y. Bartal, ``Probabilistic approximation of metric spaces and its algorithmic applications'',
In {\em Proceedings of the 37th Annual IEEE Symposium on Foundations of Computer Science},
1996, pp. 184-193.

\bibitem{Bartal98}
Y. Bartal, ``On approximating arbitrary metrics by tree metrics'',
In {\em Proceedings of the 30th Annual ACM Symposium on Theory of Computing},
1998, pp. 161-168.

\bibitem{Baker94}
B. Baker, ``Approximation Algorithms for NP-complete Problems on Planar Graphs'',
{\em J. ACM}, 41 (1994), pp. 153-180.

\bibitem{Bodlaender98}
H. Bodlaender, ``A partial k-arboretum of graphs with bounded treewidth'',
{\em Theoretical Computer Science}, 209 (1998), pp. 1-45.

\bibitem{cg:minor04}
D.E. Carroll and A.~Goel.
\newblock Lower bounds for embedding into distributions over excluded minor
  graph families.
\newblock {\em Lecture Notes in Computer Science (proceedings of the 12th
  European Symposium on Algorithms, Sep 2004)}, 3221:146--156.

\bibitem{cjlv:minors08}
A.~Chakrabarti, A.~Jaffe, J.~Lee, and J.~Vincent.
\newblock Embeddings, cuts, and flows in topological graphs: Lossy invariants,
  linearization, and 2-sums.
\newblock {\em To appear in the proceedings of the Synposium on Foundations of
  Computer Science (FOCS)}, 2008.

\bibitem{Diestel97}
R. Diestel, ``Graph Theory'', Springer-Verlag, New York, 1997.

\bibitem{FakcharRT03}
J. Fakcharoenphol, S. Rao, and K. Talwar, ``A Tight Bound on Approximating Arbitrary Metrics
by Tree Metrics'',
In {\em Proceedings of the 35th Annual ACM Symposium on Theory of Computing},
2003, pp. 448-455.

\bibitem{GuptaNRS99}
A. Gupta, I. Newman, Y. Rabinovich, and A. Sinclair,
``Cuts, trees and $\ell_1$-embeddings.'',
In {\em Proceedings of the 40th Annual IEEE Symposium on Foundations of Computer Science},
 1999, pp. 399-408.
 
\bibitem{ChekuriGNRS03}
C. Chekuri, A. Gupta, I. Newman, Y. Rabinovich, and A. Sinclair,
``Embedding $k$-Outerplanar Graphs into $\ell_1$'',
In {\em Proceedings of the 14th Annual ACM-SIAM Symposium on Discrete Algorithms},
 2003, pp. 527-536.

\bibitem{ImaseW91}
M. Imase and B. Waxman, ``Dynamic Steiner Tree Problem.'',
{\em SIAM J. Discrete Math.}, 4 (1991), pp. 369-384.

\bibitem{KleinPR93}
P. Klein, S. Plotkin, and S. Rao,
``Excluded Minors, Network Decomposition, and Multicommodity Flow'',
In {\em Proceedings of the 25th Annual ACM Symposium on Theory of Computing},
 1993, pp. 682-690.

\bibitem{Kloks94}
T. Kloks, ``Treewidth: Computations and Approximations'',
{\em Lecture Notes in Computer Science}, Vol. 842, Springer-Verlag, Berlin, 1994.

\bibitem{LinialLR95}
N. Linial, E. London, and Y. Rabinovich,
``The geometry of graphs and some of its algorithmic applications'',
{\em Combinatorica}, 15 (1995), pp. 215-245.

\bibitem{Rabinovich03}
Y. Rabinovich, ``On Average Distortion of Embedding Metrics into the Line and into $\ell_1$'',
In {\em Proceedings of the 35th Annual ACM Symposium on Theory of Computing},
2003, pp. 456-462.

\bibitem{RabinovichR98}
Y. Rabinovich and R. Raz,
``Lower Bounds on the Distortion of Embedding Finite Metric Spaces in Graphs'',
{\em Discrete \& Computational Geometry}, 19 (1998), pp. 79-94.

\bibitem{Rao99}
S. Rao, ``Small distortion and volume preserving embeddings for Planar and Euclidean metrics'',
In {\em Proceedings of the 15th Annual Symposium on Computational Geometry}, 1999, pp. 300-306.

\end{thebibliography}
\end{document}